\documentclass[a4paper, autoref, 11pt]{article}

\usepackage{comment} 
\usepackage{lipsum} 
\usepackage{fullpage} 
\usepackage{amsmath}
\usepackage{amsfonts}
\usepackage{amssymb}
\usepackage{multicol}
\usepackage[english]{babel}
\usepackage{graphicx}
\usepackage{mathtools}          
\usepackage{algorithm}
\usepackage{amsmath}
\usepackage{algpseudocode}
\usepackage{subcaption}
\usepackage{amsfonts}
\usepackage{cleveref}
\usepackage[numbers]{natbib}
\usepackage{amssymb}
\usepackage{nicefrac}
\usepackage{amsthm}
\usepackage{xcolor}
\usepackage{nicefrac}
\def \M {{\mathcal{M}}}
\def \S {{\mathcal{S}}}
\def \O {{\mathcal{O}}}

\newtheorem{theorem}{Theorem}
\newtheorem{claim}{Claim}
\newtheorem*{lemmanonumber}{Lemma}
\newtheorem{definition}{Definition}
\newtheorem{lemma}{Lemma}
\newtheorem{corollary}{Corollary}[theorem]

\makeatletter
\def\BState{\State\hskip-\ALG@thistlm}
\makeatother

\begin{document}

\title{Online and Offline Greedy Algorithms for Routing with Switching Costs}

\author{
  Roy Schwartz \\
  Technion –- Israel Institute of Technology\\
  \texttt{schwartz@cs.technion.ac.il}
  \and
 Mohit Singh \\
 Georgia Institute of Technology\\
  \texttt{mohit.singh@isye.gatech.edu}
  \and
  Sina Yazdanbod\\
  Georgia Institute of Technology\\
  \texttt{syazdanbod@gatech.edu}
}

\maketitle

\begin{abstract}
Motivated by the use of high speed circuit switches in large scale data centers, we consider the problem of {\em circuit switch scheduling}.
In this problem we are given demands between pairs of servers and the goal is to schedule at every time step a matching between the servers while maximizing the total satisfied demand over time.
The crux of this scheduling problem is that once one shifts from one matching to a different one a fixed delay $\delta$ is incurred during which no data can be transmitted.

For the offline version of the problem we present a $(1-\nicefrac[]{1}{e}-\epsilon)$ approximation ratio (for any constant $\epsilon >0$).
Since the natural linear programming relaxation for the problem has an unbounded integrality gap, we adopt a hybrid approach that combines the combinatorial greedy with randomized rounding of a different suitable linear program.
For the online version of the problem we present a (bi-criteria) $ ((e-1)/(2e-1)-\epsilon)$-competitive ratio (for any constant $\epsilon >0$ ) that exceeds time by an additive factor of $O(\nicefrac[]{\delta}{\epsilon})$.
We note that no uni-criteria online algorithm is possible.
Surprisingly, we obtain the result by reducing the online version to the offline one.

\end{abstract}

\section{Introduction}
In recent years the vast scaling up of data centers is fueled by applications such as cloud computing and large-scale data analytics.
Such computational tasks, which are performed in a data center, are distributed in nature and are spread over thousands of servers.
Thus, it is no surprise that designing better and efficient switching algorithms is a key ingredient in obtaining better use of networking resources.
Recently, several works have focused on high speed optical circuit switches that have moving optical mirrors \cite{CSSRXZWC14,tech:optical2,tech:optical3} or wireless circuits \cite{tech:wireless2,tech:wireless1,tech:wireless3}.


A common feature of many of these new switching models is that at any time the data can be transmitted on any matching between the senders and the receivers.
However, once the switching algorithm decides to reconfigure from the current matching to a new different matching, due to physical limitations such as the time it takes to rotate mirrors, a fixed delay is incurred before data can be sent along the new reconfigured matching.
This has led to significant study on obtaining good scheduling algorithms that take this delay into account~\cite{indirect:1,lance,bojjacostly}.
The cost in switching between matchings makes the problem different when compared to the classical literature on scheduling in crossbar switching~\cite{chang1999service}, which are usually based on Birkhoff von-Neumann decompositions.
In this paper we focus on finding the schedule that sends as much data as possible in a fixed time window.
We aim to design simple and efficient offline and online algorithms, with provable guarantees, for the scheduling problem that incorporates switching delays.

In the circuit switch scheduling problem, we are given a traffic demand matrix $D \in \mathbb{R}^{|A|\times |B|}_{+}$, where $A$ is the set of senders and $B$ is the set of receivers. $D_{ij}$ denotes the amount of data that needs to be sent from sender $i$ to receiver $j$. The $D_{ij}$'s can also be seen as weights on the edges of a complete bipartite graph with vertex set $A\cup B$. We are also given a time window $W$ and a switching time $\delta >0$. At any time, the algorithm must pick a matching $M$ and duration $\alpha$ for which the data is transmitted along the edges of the matching $M$ that still require data to be sent. When the algorithm changes to another matching $M'$ for another duration $\alpha'$, the algorithm must account for $\delta$ amount of time for switching between the two matchings. The total amount of time that data is sent along matchings as well as switching time between the matchings must total no more than $W$.
The objective is to maximize the total demand that is satisfied.

\subsection{Our Results and Contributions}
Our main contribution in this paper are simple and efficient algorithms for the offline and online variants of the circuit switch scheduling problem.
The following theorem summarizes our result for the offline setting.
\begin{theorem}\label{thm:offline}
Given any constant $\epsilon>0$, there is a polynomial time algorithm that returns a $(1-\nicefrac[]{1}{e}-\epsilon)$-approximation for the circuit switch scheduling problem.
\end{theorem}

We note that one can formulate two natural linear programming relaxation to the circuit switch scheduling problem.
The first assigns a distribution over matchings for every time, whereas the second picks configurations with the additional knapsack constraint.
Unfortunately, both have an unbounded integrality gap.
Thus, a different approach must be used.

We adopt a hybrid approach that combines greedy and rounding of a special linear program to prove the above theorem.
The former approach is employed when the switching delay $\delta$ is significantly smaller than the time window $W$, while the latter approach is employed otherwise.
It was already noted \cite{bojjacostly} that the offline variant of the circuit switch scheduling problem is a special case of maximizing a monotone submodular function given a knapsack constraint.
Unfortunately, the above reduction requires a ground set of infinite size where each element in the ground set corresponds to a matching $M$ and a duration $\alpha$.
We note that even if the durations are discretized we are still left with a ground set of exponential size that contains all matchings of the bipartite graph.
Hence, the standard tight $(1-\nicefrac[]{1}{e})$-approximation~\cite{sviridenko2004note} for maximizing a monotone submodular function given a knapsack constraint cannot be applied.
Our main technical ingredient is to show that despite the above difficulties, the hybrid approach we propose in the paper allows one to obtain the nearly optimal $(1-\nicefrac[]{1}{e})$-approximation for the problem.

%

We also consider the online variant of the problem where the data matrix is not known in advance but is revealed over time. We consider a discrete time process where at each time step, we receive a new additional data matrix that needs to be transmitted. Moreover, we can choose a matching to transmit data at any time step with the constraint that whenever we change the matching from the previous step, no data is transmitted for $\delta$ steps.
{Our main contribution is a reduction from the online variant to the offline variant.
To the best of our knowledge, such reductions with a minor loss in the guarantee are seldomly found.
}
This results in a bi-criteria algorithm since the online algorithm is allowed a slightly larger time window than the optimum.
We remark that such a bi-criteria approximation is necessary and we refer the reader to Appendix \ref{apx:BicriteriaOnline} for details.
The following theorem summarizes the above.
%
\begin{theorem}\label{thm:online}
Given a $\beta$-approximation for the offline circuit switch scheduling problem and an integer $k\geq 3$, there exists an algorithm achieving a competitive ratio of  $\left(1 - \nicefrac[]{2}{k}\right) \frac{\beta}{1 + \left( 1 - \nicefrac[]{2}{k} \right)\beta}$ for the online circuit switch scheduling problem which uses a time window of $W+k\delta$ as compared to a time window of $W$ for the optimum.
\end{theorem}
Combining Theorem~\ref{thm:offline} and Theorem~\ref{thm:online}, we have the following corollary.
\begin{corollary}
For any constant $\epsilon>0$, there exists an algorithm achieving a competitive ratio of $\left(\frac{e-1}{2e-1}-\epsilon\right)$ for the online circuit switch scheduling problem which uses a time window of $W+O\left(\nicefrac[]{\delta}{\epsilon}\right)$ as compared to a time window of $W$ for the optimum.
\end{corollary}
We note that the online algorithm in the above corollary runs in polynomial time.
If one is not interested in the running time of the algorithm, but rather interested only in coping with an unknown future, then Theorem \ref{thm:online} gives an online algorithm whose competitive ratio is $(\nicefrac[]{1}{2}-\epsilon)$ for any arbitrarily small constant $\epsilon >0$ (by assuming that the offline problem can be solved optimally, i.e., $\beta = 1$).

\subsection{Related Work}
Venkatakrishnan et. al. \cite{bojjacostly} were the first to formally introduce the offline variant of the circuit switch scheduling problem.
They focused on the special case that all entries of the data matrix are significantly small, and analyzed the greedy algorithm.
Though it is known that the greedy algorithm does not provide any worst-case  approximation guarantee for the general case of maximizing a monotone submodular function given a knapsack constraint, \cite{bojjacostly} proved that in the special case of small demand values they obtain an (almost) tight approximation guarantee.
To the best of our knowledge, our algorithm gives the best provable bound for the offline variant of the circuit switch scheduling problem.
A different related variant of the problem is when data does not have to reach its destination in one step, i.e., data can go through several different servers until it reaches its destination \cite{indirect:1,lance,bojjacostly}.

A dual approach is given by Liu et. al. \cite{solstice}, who aim to minimize the total needed time to transmit the entire demand matrix.
Since our algorithm aims to maximize the transmitted data in a time window of $W$, one can use our algorithm as a black box while optimizing over $W$.
It was proven in \cite{theory:NP} that the problem of minimizing the time needed to send all of the data is NP-Complete.
Hence, we can conclude that the circuit switch scheduling problem is also NP-Complete.

The problem of decomposing a demand matrix into matchings, i.e., the decomposition of a matrix into permutation matrices, was considered by \cite{theory:approx,dufosse2017further,kulkarni2017minimum,theory:mirrokni}.
The special cases of zero delay \cite{offline:zerodelay} and infinite delay \cite{offline:infdelay} have also been considered.
Several related, but slightly different, settings include \cite{optimumexists,offline:withdelay1,composite}.

Regarding the theoretical problem of maximizing a monotone submodular function given a knapsack constraint, Sviridenko \cite{sviridenko2004note} (building upon the work of Khuller et. al. \cite{khuller1999budgeted}) presented a tight $(1-\nicefrac[]{1}{e})$-approximation algorithm.
This tight algorithm enumerates over all subsets of elements of size at most three, and greedily extends each subset of size three, and returns the best solution found.
Deviating from the above combinatorial approach of \cite{khuller1999budgeted,sviridenko2004note}, Badanidiyuru and Vondr\'{a}k \cite{badanidiyuru2014fast} and Ene and Nguyen \cite{Alina2017} present algorithms that are based on an approach that extrapolates between continuous and discrete techniques.
Unfortunately, as previously mentioned, none of the above algorithms can be directly applied to the circuit switch problem due to the size of the ground set.

The online version of the circuit switch scheduling problem has been considered from a queuing theory prospective, with delays \cite{celik2016dynamic} and without delays \cite{georgiadis2006resource}. In these works,  guarantees are proven under the assumption that the incoming traffic is from a known distribution or i.i.d. random variables. To the best of our knowledge, the online version has not been studied from a theoretical perspective.

\section{Preliminaries}
First, let us start with a formal description of the problem.
We are given a complete bipartite graph $G=\left(A, B, E\right)$ where $A$ and $B$ are the sets of sending and receiving servers, a constant $\delta \geq 0$ and a time window $W \geq 0$. We are also given the traffic demand matrix of the graph, $D \in \mathbb{R}^{|A|\times |B|}_{+}$, where $D_{ij}$ denotes the amount of data that needs to be sent from sender $i$ to receiver $j$. The $D_{ij}$'s can be seen as weights on the edges of the complete bipartite graph. To simplify the notation, for an edge $e=(i,j)$ we abbreviate $D_{ij}$ to $D_{e}$. Let $\mathcal{M}$ be the collection of all matchings in $G$.
\begin{definition}
The pair $(M, \alpha)$ is called a \textit{configuration} if $M \in \mathcal{M}$ and $\alpha \in \mathbb{R}_{+}$.
\end{definition}
The term \textit{scheduling} a configuration $\left(M, \alpha\right)$ means sending data via the matching $M$ for a duration of time that equals $\alpha$.
For simplicity of presentation, we also interpret a matching $M$ as a $\{0,1\}^{|A|\times |B|}$ matrix where $e \in M$ if and only if the entry of edge $e$ in $M$ equals $1$.
Note that for any edge $e \in M$ 
the total data sent through $e$ would be $\min(D_{e}, \alpha)$ and the total amount of data sent by the configuration would be $||\min\left(D, \alpha M\right) ||_{1} = \sum_{e\in M} \min\left(D_{e},\alpha\right)$ (note that the minimum is taken element-wise).
For simplicity of presentation we may use $||.||_{1}$ and  $||.||$ interchangeably.

Switching from a configuration $\left(M, \alpha\right)$ to  another $\left(M^{\prime}, \alpha^{\prime}\right)$ incurs a given constant delay $\delta$, during which no transmission can be made. Let $\mathcal{C}$ denote the collection of all possible configurations.

\begin{definition}
A \textit{schedule $S$ of size $k$} is a subset $S \subseteq \mathcal{C}$ such that $|S|=k$. We say that $S$ requires a total time of $\sum_{\left(M, \alpha\right)\in S} \left(\alpha + \delta\right)$ to be scheduled.
\end{definition}
 The total time of the schedule includes both the time for sending data with each configuration and the delay in switching between them.
 This brings us to the definition of a feasible schedule.
 \begin{definition}
A schedule $S$ is \textit{feasible} if $\sum_{\alpha:(M, \alpha)\in S} (\alpha + \delta) \leq W$.
\end{definition}
 In the offline setting, the goal is to find a feasible schedule $S$ that maximizes the data sent over the given time window of length $W$. This problem can be formulated as follows:
 \begin{align}
 \max \left\{ \left|\left|\min\left(D, {\textstyle \sum_{(M,\alpha) \in S}} \alpha M\right)\right|\right| _1:S\subseteq \mathcal{C}, {\textstyle \sum_{\alpha:(M,\alpha)\in S}} \left(\alpha + \delta\right) \leq W\right\}.
 \end{align}
We note that $\mathcal{C}$ might be of infinite size.
However, we use standard discretization techniques to limit the set of possible values of $\alpha$ in our algorithms. We will discuss this with more detail in the later relevant sections. For now, assume $\mathcal{C}$ is finite. To facilitate the notation and the analysis of our problem, we turn to a well-known class of functions called \textit{submodular functions}.

\begin{definition}
Given a ground set $N=\{1,2,3,...,n\}$, a set function $f:2^{N}\rightarrow \mathbb{R}_{+}$ is a \textit{submodular function} if for every $A,B \subseteq N$: $f(A) + f(B) \geq f(A\cup B) + f(A \cap B)$.
\end{definition}
For our problem, define $f:2^{\mathcal{C}} \rightarrow \mathbb{R}_{+}$ as: $$f\left(S\right) = \left|\left|\min\left(D, {\textstyle \sum_{\left(M, \alpha\right)\in S}} \alpha M\right)\right|\right| _1 .$$
Moreover, we denote by $f_{S}\left(\left(M,\alpha\right)\right) = f\left(S \cup \left(M,\alpha\right)\right) - f\left(S\right)$ the marginal gain of the schedule $S$ if the configuration $\left(M, \alpha\right)$ was added to it.
It has been shown that $f$ is submodular (refer to Theorem $1$ in \cite{bojjacostly}).
For the sake of completeness, we state the theorem. Note that $f$ is \textit{monotone} if for every $A\subseteq B \subseteq N$: $f(A) \leq f(B)$.
\begin{theorem}[Theorem 1 in \cite{bojjacostly}]
The function $f$ is a monotone submodular function.
\end{theorem}

For the online version of the problem, we use a discrete time model. Unlike the offline version, in the online setting we do not know the entire traffic matrix of the graph in the beginning. We start with $D_{0}$ as the demand matrix already present in the initial graph.  At time $t$ an additional traffic matrix $D_{t}$ is revealed to the algorithm that includes new demands for data that need to be transmitted. In the online version of the problem sending configuration $\left(M,\alpha\right)$ means that for the next $\alpha \in \mathbb{Z}_{+}$ time steps our algorithm is busy sending the matching $M$. Switching a configuration to a different one incurs {{an additional}} delay of $\delta \in \mathbb{N}$ steps, during which no data can be sent.
The incoming traffic matrices, at every step {starting with the sending of $\left(M,\alpha\right)$ and ending with the switching cost (a total of $\alpha +\delta$ time steps)}, will accumulate {{and be added to}} the remaining traffic matrix of the graph. 

\section{Offline Circuit Switch Scheduling Problem}

In this section, we prove Theorem~\ref{thm:offline} by giving an approximation algorithm for the circuit switch scheduling problem. Our algorithm is a combination of the greedy algorithm as well as a linear programming based approach. We first show that the greedy algorithm gives close to a $(1-\frac1e)$-approximation if $\delta$, the switching time, is much smaller than the time window. This is done in Section~\ref{sec:greedy}. In Section~\ref{sec:lp}, we give a randomized rounding algorithm for a linear programming relaxation that gives a $(1-\frac1e)$-approximation but runs in time exponential in number of matchings used in the optimal solution. While the natural linear program for the problem has unbounded gap, we show how to bypass this when the schedule has a constant number of matchings.

\subsection{Greedy Algorithm}\label{sec:greedy}

The greedy algorithm is as follows: at each step choose the configuration that maximizes the amount of data it sends per unit of time it uses.
Formally, if $R_{i}$ is the remaining data demand in the graph after $i$ configurations were already chosen, the greedy algorithm will choose the following configuration to be used next:
\begin{equation}
\label{eq:greedy}
\left(M_{i+1}, \alpha_{i+1}\right) = \text{argmax}_{M \in \mathcal{M}, \alpha \in \mathbb{R}_{+}} \frac{||\min\left(R_{i}, \alpha M\right)|| _1}{\alpha + \delta}.
\end{equation}

The greedy algorithm continues to pick configurations until the first time the time constraint is violated or met. Algorithm~\ref{alg:IGA} demonstrates this process. Let $r$ denote this number of steps and $\S_{r}$ the schedule created after $r$ steps of this algorithm. The last chosen configuration may violate the time window budget and a natural strategy is to reduce its duration to the time window $W$ as is done in Step (11)-(12) of the algorithm. Indeed~\cite{bojjacostly} analyzes this algorithm and shows that it performs well if each entry in data matrix is small. They also show that the above optimization problem can be solved using the maximum weight matching problem. We give a different analysis of the algorithm and show that it gives us a $\left(1-\frac{1}{e}-\epsilon\right)$-approximation if $\delta< (\frac{e}{2(e-1)}\epsilon) \cdot W$.

\begin{theorem}\label{thm:greedy}
Let $\S_r$ denote the schedule as returned by the greedy algorithm and $\O$ denote the optimal schedule. Then
$$f(\S_r) \geq \left(1-\frac{2\delta}{W}\right) \left(1-\frac1e\right) f(\O).$$
\end{theorem}

\begin{algorithm}[t]
\caption{Greedy Algorithm}\label{alg:IGA}
\begin{algorithmic}[1]
\State \texttt{Input:} $G=\left(A,B, E\right),D, \delta, W$
\State \texttt{Output:} $\{\left(M_{1},\alpha_{1}\right), \dots, \left(M_{r}, \alpha_{r}\right)\}$
\State $\S \gets \emptyset$. $i\gets 0$, $R_1 \gets D$.
\While{$\sum_{\alpha:(M,\alpha) \in \S} \left(\alpha+\delta\right) \leq W$}
\State $i\gets i+1$.
\State $\left(M_{i}, \alpha_{i}\right) \gets \arg\max_{M \in \mathcal{M}, \alpha \in \mathbb{R}_{+} } \frac{||\min\left(R_{i}, \alpha M\right)||}{\alpha + \delta} $.
\State $\S \gets \S \cup \{\left(M_{i}, \alpha_{i}\right)\}$.
\State $R_{i+1} \gets R_{i} - \min\left(R_{i}, \alpha_{i}M_{i}\right)$.
\EndWhile
\State $r\gets i$.
\If{ $\sum_{(M,\alpha)\in \S} \left(\alpha + \delta\right) > W$ }
\State $\beta_r \gets W-\delta-\sum_{j=1}^{r-1}(\alpha_j+\delta)$
\If{ $\beta_r \geq 0$}
\State $\S \gets (\S \setminus \{(M_{r},\alpha_{r})\}) \cup \{(M_{r}, \beta_r)\}$
\Else
\State $\S \gets (\S \setminus \{(M_{r},\alpha_{r})\})$
\EndIf
\EndIf\\
\Return $\S$
\end{algorithmic}
\end{algorithm}


\begin{proof}
To analyze the algorithm, we first show that the objective of the optimal schedule of a slightly smaller time window $W-\delta$ is not much smaller than the optimum value of the optimum schedule for time window $W$ in Lemma~\ref{lem:reduce}. Indeed, the lemma states that given any schedule for time window $W$, for example the optimal schedule, there exists a schedule with time window $W-\delta$ of a comparable objective.

\begin{lemma}\label{lem:reduce}
For any schedule $\S$ for a time window of $W$, there is a schedule $\Tilde{\S}$ on a window of $W - \delta$ time such that $f(\Tilde{\S}) \geq \left(1-\frac{2\delta}{W}\right)f\left(\S\right)$.
\end{lemma}
\begin{proof}
Let $T_{\textrm{data}}$ be the total time spent sending data and $T_{\textrm{switch}}$ be the total time spent switching between configurations. Thus, $W = T_{\textrm{data}} + T_{\textrm{switch}}$. We prove that we can remove $\delta$ time from some configuration or we can remove an entire configuration from $\S$ while reducing the objective by no more than $\frac{2\delta}{W}$ fraction of the objective. Consider the two following cases for the given $\S$. If $T_{\textrm{data}}\geq \frac{W}{2}$, we have $\frac{f\left(\S\right)}{T_{\textrm{data}}} \leq \frac{2}{W}f\left(\S\right)$. Thus there exists a configuration that we can deduct $\delta$ time from and at most lose $\frac{2 \delta}{W}f\left(\S\right)$. If $T_{\textrm{switch}} \geq \frac{W}{2}$. This means the number of configurations is at least $\frac{W}{2\delta}$. Each configuration on average sends $\frac{2\delta}{W}f\left(\S\right)$ data. Therefore, there is a configuration we can completely remove from our schedule such that total amount of lost data is at most $\frac{2\delta}{W}f\left(\S\right)$. In both cases we can reduce the time taken by the schedule by at least $\delta$ and have a new schedule $\Tilde{\S}$ such that $f\left(\Tilde{\S}\right) \geq \left(1 - \frac{2\delta}{W}\right)f\left(\S\right)$.
\end{proof}

Let $\O'$ denote the optimal solution with time window $W-\delta$. From Lemma~\ref{lem:reduce}, we have
$f(\O')\geq \left(1-\frac{2\delta}{W}\right) f(\O)$. In the following lemma, we show that the output of the greedy algorithm is at least a $\left(1-\frac{1}{e}\right)$-approximation of $f(\O')$. The proof of the lemma follows standard analysis for greedy algorithms for coverage functions, or more generally submodular functions, except for the being careful at the last step. The proof appears in Appendix \ref{apx:MissingProofs}. The proof of Theorem~\ref{thm:greedy} now follows immediately.

 \begin{lemma}\label{lem:beta}
If $\O'$ is the optimum schedule on time window $W - \delta$, then $$f(\S_r)\geq (1-\frac1e) f(\O').$$
\end{lemma}

\end{proof}
\subsection{Linear Programming Approach for Constant Number of Configurations}\label{sec:lp}
In this section, we assume that we want to schedule at most a given constant $k$ number of configurations and prove the following theorem.

\begin{theorem}\label{thm:constant}
There exists a randomized polynomial time algorithm that given an integer $k$ and an instance of the circuit switch scheduling problem returns a feasible schedule whose objective, in expectation, is at least $(1-\frac1e-\epsilon)$ of the optimum solution that uses at most $k$ matchings. Moreover the running time of the algorithm is polynomial in $\frac{n}{\epsilon^k}$.

\end{theorem}

 Let us denote optimum schedule by $\O=\{(M_{1}^{*}, \alpha_{1}^{*}), \dots, (M_{k}^{*}, \alpha_{k}^{*})\}$. Note that, without the loss of generality, we can assume that we know what the $\alpha^{*}_{i}$'s are. This can be done by a standard discretization of the possible values. Since, the number of configurations is constant this enumeration will be polynomial in $\frac{1}{\epsilon^k}$ to an accuracy of $\epsilon$. The total data sent by a schedule $\S$ is $f(\S)=||\min(D, \sum_{(M,\alpha)\in \S} \alpha M)||_{1}$. However, in this section, it is more beneficial to consider the total data as the sum of total data sent over each edge. We model the total data by $Z = \sum_{e \in E} z_{e}$, where $z_{e}$ is the amount of data that was sent through edge $e$ in our graph. In the case of the optimum, $z^*_{e} = \min(D_e, \sum_{\alpha^{*}:(M^{*},\alpha^{*})\in \O: e \in M^{*}} \alpha^{*})$ and $Z^* = \sum_{e \in E} z^*_{e}$. We can formulate the following integer program for this problem.
\begin{align}
(\mathcal{P})~~~~~\max~~~~~ & \sum _{e\in E} z_e & \label{LP:objective}\\
~~~~~ s.t.~~~~~ & \sum _{M\in \mathcal{M}}x_{M,i}\leq 1 & \forall i=1,\ldots,k \label{LP:constraint1}\\
& z_e \leq D_e & \forall e\in E \label{LP:constraint2}\\
& z_e \leq \sum _{i=1}^{k}\sum _{M\in \mathcal{M}:e\in M}\alpha^*_i \cdot x_{M,i}& \forall e\in E \label{LP:constraint3}\\
& x_{M,i} \in \{0,1\} & \forall e\in E, \forall M\in \mathcal{M}, \forall i=1,\ldots,k \nonumber
\end{align}

Constraints \eqref{LP:constraint1} is to ensure that only one matching is considered in every time interval. Constraint \eqref{LP:constraint2} and \eqref{LP:constraint3} are to model the total data sent. We can relax this integer program to an LP by changing the $x_{M,i} \in \{0,1\}$ to $0 \leq x_{M,i} \leq 1$. The following lemma states that the relaxed linear program is a relaxation of our problem for the constant number of configurations.

\begin{lemma}\label{lem:relaxedLP}
Let $Z_{LP}$ be the value of an optimum solution to the LP, then $Z_{LP} \geq Z^{*}$
\end{lemma}
\begin{proof}
If $\O=\{(M_{1}^{*}, \alpha_{1}^{*}), \dots, (M_{k}^{*}, \alpha_{k}^{*})\}$ is our optimum answer, based on $\O$ we will create a feasible answer to the LP. For every $(M^{*}_{i}, \alpha^{*}_{i}) \in \O$, we set $x_{M^{*},i}=1$. Clearly, the constraint \ref{LP:constraint1} is satisfied since we picked exactly one matching for every interval. The constraints \ref{LP:constraint2} and \ref{LP:constraint3} is by definition satisfied since $f(\O) = ||\min(D, \sum_{(M, \alpha) \in \O} \alpha M)||$ and the constraints are modeling this minimum. This argument shows that the optimum answer is feasible in the LP and since the LP is a maximization problem we can conclude that $Z_{LP}\geq f(\O)$.
\end{proof}
The LP contains an exponential number of variables, since the number of matchings in the complete graph is exponential in the size of the graph. To be able to solve this program we need to introduce a separation oracle for the dual of this LP. The following program is the dual of our LP.
\begin{align}
(\mathcal{D})~~~~~\min~~~~~ & \sum _{i=1}^{k} y_i + \sum _{e\in E}d_e a_e & \label{DUAL:objective}\\
~~~~~ s.t. ~~~~~ & y_i \geq \alpha_i^* \sum _{e\in M}b_e & \forall M\in \mathcal{M}, \forall i=1,\ldots,k \label{DUAL:constraint1}\\
& a_e+b_e \geq 1 & \forall e\in E \label{DUAL:constraint2}\\
& a_e\geq 0, ~ b_e\geq 0, ~ y_i\geq 0 & \forall e\in E, \forall i=1,\ldots,k \nonumber
\end{align}
The Lemma \ref{lem:seporacle} states the existence of a separation oracle.
\begin{lemma}\label{lem:seporacle}
The dual program $\mathcal{D}$ admits a polynomial time separation oracle.
\end{lemma}
\begin{proof}
Given a solution $(\{ y_i\}_{i=1}^{k},\{ a_E\}_{e\in E}, \{ b_e\} _{e\in E})$ we are required to determine whether it is feasible and if not provide a constraint that is violated. We can easily determine whether all constraints of type (\ref{DUAL:constraint2}) are satisfied, and if not provide one that is violated, by a simple enumeration over all edges $e\in E$. The same can be done for constraints of type (\ref{DUAL:constraint1}) by enumerating over $i=1,\ldots,k$ and for each $i$ compute a maximum weight matching in $G$ equipped with $\{ b_e\}_{e\in E}$ as edge weights and check whether the maximum weight matching has value at most $y_i/\alpha_i^*$. If the maximum weight matching exceeds the target value return the constraint that corresponds to $i$ and the maximum weight matching.
\end{proof}
Solving the linear program will provide us with a fractional solution $\{x_{M,i}\}_{M \in \mathcal{M},i=1,\dots,k}$. For any $i$ we have $\sum_{M \in \mathcal{M}} x_{M,i} \leq 1$. This constraint of the LP creates a distribution over the matchings in time interval $i$. We create a solution to the program $\mathcal{P}$ from the fractional solution by a randomized rounding technique. We pick $M \in \mathcal{M}$ for the time interval $i$ with probability $x_{M,i}$. Note that with probability $1 - \sum_{M \in \mathcal{M}} x_{M,i}$ no matching will be chosen for this time interval. A formal description of this rounding method is provided in Algorithm \ref{alg:rounding}. Let $X_{M,i}$ denote the indicator random variable if matching $M$ is selected for the $i^{th}$ slot. Moreover, let $Y_{e,i}$ denote the random variable that edge $e$ is present in the matching chosen in the $i^{th}$ slot. We have $Y_{e,i}=\sum_{M\in \M: e\in M} X_{M,i}$ for each $e\in E$ and $i$ and $E[Y_{e,i}]=\sum_{M\in \M: e\in M} x_{M,i}$. Moreover, let $Z_e$ denote the random variable that denotes the data sent along edge $e$. Then we have
$Z_e= \min(D_{e}, \sum_{i=1}^k\alpha^*_i Y_{e,i})$. Observe that the random variables $\{Y_{e,i}\}_{i=1}^k$ are independent.

\begin{algorithm}[t]
\caption{Randomized Rounding}\label{alg:rounding}
\begin{algorithmic}[1]
\State \texttt{Input:} $(k,\{ \alpha_i^*\}_{i=1}^{k},\{ x_{M,i}\}_{M\in \mathcal{M},i=1,\ldots,k})$
\State \texttt{Output:} $\{ (M_i,\alpha_i^*)\}_{i=1}^{k}$
\For{\texttt{$i \gets 1,\ldots,k$}}
\State choose $M_i$ to be a random matching w.p. $x_{M,i}$ for the interval $i$
\EndFor\\
\Return $\{ (M_i,\alpha_i^*)\}_{i=1}^{k}$
\end{algorithmic}
\end{algorithm}

The following Lemma~\ref{lem:rounding} is implicit in Theorem 4 of Andelman and Mansour~\cite{andelman04auctions}.

\begin{lemma}\label{lem:rounding}
Let $Y_1,\ldots, Y_n$ be independent Bernoulli random variables and let $Z=\min(B,\sum_{i=1}^n b_i Y_i)$ for some non-negative reals $B, b_1,\ldots, b_n$. Then
   $$\mathbb{E}[Z] \geq \left(1-\frac{1}{e}\right)\min\left(B, \mathbb{E}\left[\sum_{i=1}^n b_i Y_{i}\right]\right).$$
\end{lemma}

Applying the above lemma for each $e$ and random variables $\{Y_{e,i}\}_{i=1}^k$, we obtain that

\begin{eqnarray*}
\mathbb{E}[Z_e]&\geq&  \left(1-\frac1e\right) \min\left(D_e, \mathbb{E}\left[\sum_{i=1}^k \alpha^*_i Y_{e,i}\right]\right)=\left(1-\frac1e\right) \min\left(D_e, \sum_{i=1}^k \sum_{M\in \M: e\in M}\alpha^*_i x_{e,i}\right)\\
&\geq&\left(1-\frac1e\right)\cdot z_e.
\end{eqnarray*}

Now summing over all edges, Theorem~\ref{thm:constant} follows. We are now ready to conclude our discussion of the offline variant of the circuit switch scheduling problem and prove Theorem~\ref{thm:offline}.
\begin{proof}[Proof of Theorem \ref{thm:offline}]
Given $\epsilon > 0$, if $ \delta \leq (\frac{e}{2(e-1)}\epsilon)W$ then Theorem~\ref{thm:greedy} gives us a $(1 - \frac{1}{e} - \epsilon)$-approximation. Otherwise, $\frac{2(e-1)}{e}\frac{1}{\epsilon} > \frac{W}{\delta}$ implying that at most $\frac{2(e-1)}{e}\frac{1}{\epsilon}$ configurations can be scheduled. In this case, Theorem~\ref{thm:constant} will give a $(1 - \frac{1}{e} - \epsilon)$-approximation.
\end{proof}

\section{Online Circuit Switch Scheduling Problem}
In this section, we prove Theorem~\ref{thm:online}. Recall that in the online setting, we consider a discrete time model\footnote{We could also consider a continuous time model where data matrices can arrive at any time and the algorithm can choose a matching at any time instant with a switching time $\delta$ when no data is sent. Our results apply to this model as well. The discrete model makes the presentation of the results easier.} where an additional traffic matrix is revealed at every time $t=1,2,\ldots, T$. At every time step $t$, a new set of traffic demands arrives and adds to the remaining traffic that has not been sent so far. We assume that the data matrix arriving at each step is integral and thus can be modeled as a multigraph. We denote the incoming traffic matrices as multigraphs $\{E_{1}, E_{2}, \dots, E_T \}$ (instead of $D_{i}$'s to simplify and familiarize the notation) and thus union of any two such graphs is defined by adding the number of copies of edges in the two constituents. Before proving the general theorem, we first consider the case when there is no delay while switching matchings, i.e., $\delta=0$. Observe that in this case, the offline problem can be solved exactly and we show a $\frac12$-competitive algorithm for the online problem. The general reduction builds on this simple case along with the offline algorithm.

\subsection{Without Configuration Delay}
Observe that an online algorithm, in this case, will pick a set of matchings $\{M_{1}, M_{2}, \dots, M_T\}$, instead of a schedule, that covers the maximum number of edges. At each step $t$, the algorithm picks the maximum matching from the graph formed by the new edges that arrive, $E_{t}$, and the remaining edges in the graph from previous steps which we denote by $R_{t-1}$. The algorithm is formally given in Algorithm \ref{alg:onlinewdelay}. Here $\M$ denotes the set of all matchings on the complete bipartite graph with parts $A$ and $B$. The objective of Algorithm~\ref{alg:onlinewdelay} is $\sum_{t = 1}^{T} |M_{t}|$, where $|M_{t}|$ denotes the number of edges in the matching $M_{t}$. We denote the optimum solution by $\O=\{O_{1}, \dots, O_{T} \}$, We have the Theorem \ref{thm:onlinewdelay} for our approximation guarantee.

\begin{algorithm}
\caption{Online Greedy Algorithm without Delay}\label{alg:onlinewdelay}
\begin{algorithmic}[1]
\State \texttt{Input:} Bipartite multigraphs on $E_{1}, E_{2},\dots, E_T$ on $A\cup B$ where $E_t$ is disclosed at beginning of step $t$.
\State \texttt{Output:} $\{M_{1}, M_{2}, \dots, M_T\}$
\State $R_{0}, S \gets \emptyset$, $t \gets 1 $.
\For{$t \gets 1,2,\dots, T$}
\State $R_t'\gets R_{t-1}\cup E_t$.
\State $M_{t} \gets argmax_{M \in \mathcal{M}, M\subseteq R_{t}'} |M|$.
\State $S \gets S \cup \{M_{t}\}$, $R_{t} \gets R_t'\setminus \{M_{t}\}$, $t \gets t + 1$.
\EndFor\\
\Return $S$
\end{algorithmic}
\end{algorithm}

\begin{theorem}\label{thm:onlinewdelay}
Algorithm \ref{alg:onlinewdelay} is $\frac{1}{2}$-competitive for the online circuit switch scheduling problem without delays.
\end{theorem}
\begin{proof}
Let $\Gamma = \{E_{1}, \ldots, E_{T}\}$ denote the incoming edges for the first $T$ steps. We call this the input sequence for the first $T$ steps. We use induction on $T$ to prove the theorem. Specifically, we prove that for \emph{any} input sequence of edges for $T$ steps, $\Gamma = \{E_{1}, E_{2}, \dots, E_{T}\}$, we have
$$\sum_{t=1}^{T} |M_{t}| \geq \frac{1}{2}\sum_{t=1}^{T} |O_{t}|$$
For $T=1$, we know that the maximum matching has the biggest size of any matching in the graph. So, we have $|M_{1}| \geq |O_{1}|$ and thus the base case holds.

By the induction hypothesis, we have that for \emph{any} input sequence of $T-1$ steps, we have $\sum_{t=1}^{T-1} |M_{t}| \geq \frac{1}{2} \sum_{t=1}^{T-1} |O_{t}|$ where $\{M_t\}_{t=1}^{T-1}$ and $\{O_t\}_{t=1}^{T-1}$ are the output of the algorithm and the optimal solution, respectively.

Now, consider any input sequence $E_1,\ldots, E_T$. Recall, $R_1$ is the residual graph formed after first step of the algorithm, i.e. $R_1=E_1\setminus M_1$. At the next step, the algorithm will find the maximum matching in $R_2'=R_{1} \cup E_{2}$ as its edge set. We build a new sequence of $T-1$ inputs and apply induction to it.

Let $\Gamma^{\prime} =\{R_{2}', E_{3}, \dots, E_{T}\}$. Consider the optimum solution on this new input sequence. Let $\{M_{t}^{\prime}\}_{t=2}^T$ be the matchings that our algorithm picks given this new input sequence and $\{O_{t}^{\prime}\}_{t=2}^T$ the optimum matchings. Using the induction hypothesis we can write
$\sum_{t=2}^{T} |M^{\prime}_{t}| \geq \frac{1}{2} \sum_{t=2}^{T} |O^{\prime}_{t}|.$

First note that for $2\leq i \leq n, M_{i} = M_{i}^{\prime}$. This is true since $M_{i}$ and $M_{i}^{\prime}$ are the maximum matchings of the same graph as can be seen inductively. We now show the following lemma that relates the optimum solution of the new instance to the original instance.
\begin{lemma}\label{lem:main}
$\sum_{t=2}^{T} |O^{\prime}_{t}| \geq \sum_{t=2}^{T} |O_{t}| - |M_{1}|. $
\end{lemma}
\begin{proof}
The matchings $\{O_{2}\setminus M_{1}, O_{3}\setminus M_{1}, \dots, O_{T}\setminus M_{1} \}$ is a feasible output for the optimum solution on the $\Gamma^{\prime}$ sequence. Therefore, we have
$ \sum_{t=2}^{T} |O_{t}^{\prime}| \geq \sum_{t=2}^{T} |O_{t}| - |M_{1}| $
as required.
\end{proof}
Using the induction hypothesis and the lemma we can write
$$ \sum_{t=2}^{T} |M_{t}| \geq \frac{1}{2}\left(\sum_{t=2}^{T} |O_{t}| - |M_{1}|\right)$$
Adding the inequality $|M_{1}| \geq |O_{1}|$ to both sides, we obtain
$$ \sum_{t=1}^{T} |M_{t}| \geq \frac{1}{2}\left(\sum_{t=2}^{T} |O_{t}|\right) +\frac{1}{2}|M_1|\geq \frac{1}{2}\left(\sum_{t=2}^{T} |O_{t}|\right) + \frac12 |O_1|=\frac{1}{2}\left(\sum_{t=1}^{T} |O_{t}|\right) $$
and the induction step follows.\end{proof}

\subsection{With Configuration Delay}
In this section, we assume switching between the configurations causes a delay of $\delta \in \mathbb{N}$ steps during which no data is sent. We also assume that we have access to a $\beta$-approximation for the offline version of the problem. Note that we view the offline algorithm as a black-box. More formally, we assume we have an algorithm of the form Algorithm~\ref{alg:offlinebox}. To reiterate, $G$ is the given complete bipartite graph, $D$ is the traffic demand matrix, $\delta$ is the switching delay and $W$ is the size of the time window. Recall, that sending the configuration $(M, \alpha)$ means that for the next $\alpha$ steps we will only send data using matching $M$.

\begin{algorithm}
\caption{Offline Algorithm for Circuit Switch Scheduling}\label{alg:offlinebox}
\begin{algorithmic}[1]
\State \texttt{Input:} $G=\left( A, B, E \right), D, \delta,  W$
\State \texttt{Output:} $\S = \{(M_{1}, \alpha_{1}), \dots, (M_{j}, \alpha_{j})\}$
\end{algorithmic}
\end{algorithm}

Given a constant $k\geq 1$, the first step of the algorithm is to wait $k\delta$ steps for data to accumulate and then run the offline algorithm on the accumulated data for time window $W = k\delta$.
Let $\S_{1}$ be the output of the offline algorithm. We run this schedule from time $t=k\delta + 1$ to $t = 2k\delta$. Meanwhile, we collect the incoming data matrices in these times. Figure \ref{fig:basis} shows one step of the algorithm.  At the next step, we consider the total remaining data that includes data that has not been scheduled so far from previous schedule(s) and newly arrived data in previous $k\delta$ steps. We then run the offline algorithm on this data matrix to obtain a schedule for the next $k\delta$ steps. More generally, we continue this process for every block of $k\delta$ time steps. Algorithm \ref{alg:onlinedelay} is the formal description of the algorithm. Note that this description is written as an enumeration over blocks of size $k\delta$. Recall that $f(\S)$ denotes the amount of data sent by any schedule $\S$.

\begin{algorithm}
\caption{Online Greedy with Delay}\label{alg:onlinedelay}
\begin{algorithmic}[1]
\State \texttt{Input:}$\delta, k$ and data matrices $D_{1},D_{2},\dots,D_T$ on $A\times B$ where $D_i$ revealed at beginning of step $i$. Let $l=\lceil \frac{T}{k\delta}\rceil$.
\State \texttt{Output:}$\mathcal{S}=\S_{1} \cup \S_{2} \cup \ldots \cup \S_l$.
\State $\S \gets \emptyset$, $R_{0} \gets \emptyset$.
\For{ $r \gets 0, \dots, l-1$ }
\State $R_{r}' \gets R_{r} + \sum_{rk\delta + 1\leq j \leq \left(r+1\right)k\delta} D_{j}$.
\State $\S_{r} \gets OfflineAlgorithm\left(G, R_{r}', \delta, k\delta\right)$.
\State $R_{r+1} \gets R_{r}' - min\left(R_{r}', \sum_{\left(\alpha,M\right)\in \S_r} \alpha M\right)$,  $\S \gets {\S} \cup \S_{r}$.
\EndFor\\
\Return $\S$
\end{algorithmic}
\end{algorithm}
\begin{proof}[Proof of Theorem \ref{thm:online}] 
We use a coefficient $\gamma \leq \beta$ and optimize $\gamma$ in the end. We prove the theorem by induction on the number of the blocks, i.e., $l$ and will follow along the lines of proof of Theorem~\ref{thm:onlinewdelay}. As we did in the proof of Theorem \ref{thm:onlinewdelay}, we consider the incoming traffic as sequences. But in this case we define a sequence $\Gamma =\{I_{1}, I_{2}, \dots, I_{l}\}$, where $I_{i} = \bigcup_{j=\left(i-1\right)\left(k\delta\right)+1}^{i \left(k \delta \right) } D_{j}$ is the input of block $i$. For $l = 1$, let the optimum schedule be $\O$ and the algorithm's schedule be $\S$. Figure \ref{fig:basis} shows this setting.
\begin{figure}[h]
\centering
\includegraphics[width=0.35\textwidth]{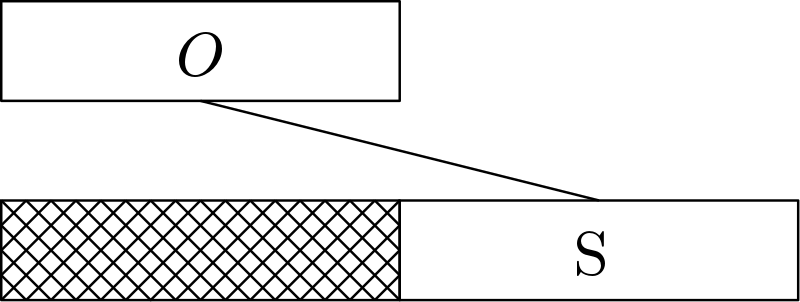}
\caption{Basis of the induction. The crossed out block is the waiting period of our algorithm}
\label{fig:basis}
\end{figure}
Using Lemma \ref{lem:reduce}, there exists a schedule $\Tilde{\O}$ with the property that $f\left(\Tilde{\O}\right) \geq \left(1 - \frac{2}{k}\right)f\left(\O\right)$. Since $\S$ is the output of our offline algorithm we can write $f\left(\S\right) \geq \beta f\left(\O^{\prime}\right) \geq \left(1 - \frac{2}{k}\right)\beta f\left(\O\right) \geq \left(1 - \frac{2}{k}\right)\gamma f\left(\O\right)$ and the basis of the induction is proven.

For $l = t$, again let $\O$ be the optimum schedule and $\S = S_{1} \cup S_{2} \dots\cup S_{t}$ be the output of our algorithm where each $S_{i}$ is the schedule on $i$th $k\delta$ block. Let $O_{1}$ be the optimum schedule for the first block and $S_{1}$ our algorithm's schedule on that block. Refer to Figure \ref{fig:step} for an illustration of this setting.

\begin{figure}[h]
\centering
\includegraphics[width=0.85\textwidth]{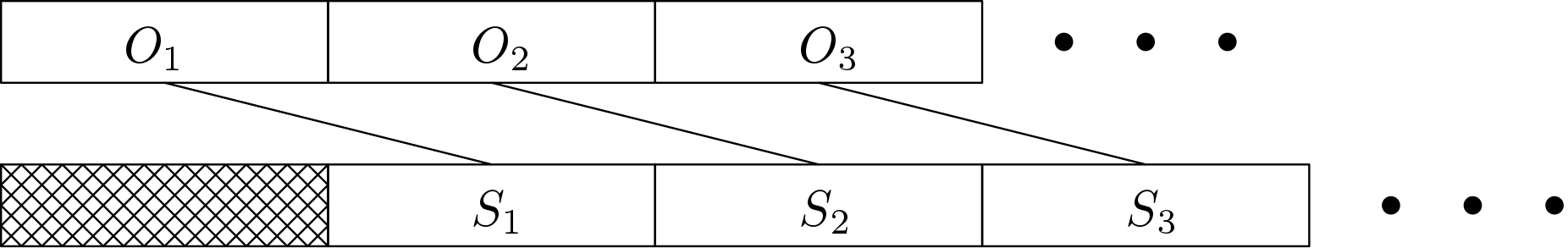}
\caption{Step of the induction.}
\label{fig:step}
\end{figure}

Consider the new input sequence $\Gamma^{\prime} = \{R_{1}' \cup I_{2}, I_{3}, \dots, I_{t}\}$. Let the optimum schedule on the new input sequence be $\O^{\prime}$ and the algorithm's schedule be $\S^{\prime} = S_{1}^{\prime} \cup \dots \cup S_{l}^{\prime}$. From the induction hypothesis, we have
$$f\left(\S^{\prime}\right) \geq \left(1 - \frac{2}{k}\right)\gamma f\left(\O^{\prime}\right).$$
Note that $S_{i} = S^{\prime}_{i-1}$ for $i \geq 2$ and thus $$f\left(\S^{\prime}\right) = f\left(\S\setminus S_{1}\right) = f\left(\S\right) - f\left(S_{1}\right).$$ As in the proof of Lemma~\ref{lem:main}, a candidate schedule for the new instance is to consider $\O\setminus O_1$ and  ignore the data sent by the algorithm in the schedule $S_1$ if it appears in any of the optimal matchings. Thus we obtain that  $$f\left(\O^{\prime}\right) \geq f\left(\O\setminus O_{1}\right) - f\left(S_{1}\right) = f\left(\O\right) - f\left(O_{1}\right) - f\left(S_{1}\right).$$ For $O_{1}$ based on our basis argument we can find $S_{1}$ such that $f\left(S_{1}\right) \geq \left(1 - \frac{2}{k}\right)\beta f\left(O_{1}\right) $. To sum up, we have the two following inequalities:
$$ f\left(\S\right) - f\left(S_{1}\right) \geq \left(1-\frac{2}{k}\right)\gamma\left(\left(f\left(\O\right) - f\left(O_{1}\right)\right) - f\left(S_{1}\right)\right) $$
and
$$ f\left(S_{1}\right) \geq \left(1 - \frac{2}{k}\right)\beta f\left(O_{1}\right). $$
Rewriting the first inequality, we have
$$f\left(\S\right) - \left(1 - \left(1 - \frac{2}{k}\right)\gamma\right)f\left(S_{1}\right) \geq \left(1 - \frac{2}{k}\right)\gamma\left(f\left(\O\right) - f\left(O_{1}\right)\right) $$
Adding the $\left(1 - \left( 1 - \frac{2}{k}\right)\gamma\right)$ times the second inequality
$$f\left(\S\right) \geq \left(1 - \frac{2}{k}\right)\gamma f\left(\O\right) - \left( 1 - \frac{2}{k}\right) \left(\gamma - \beta \left(1 - \left( 1 - \frac{2}{k}\right)\gamma \right)\right)f(O_{1}) $$
Optimizing the $\gamma$ we get $\gamma = \frac{\beta}{\left( 1 + \left(1 - \frac{2}{k}\right)\beta \right)}$ and thus proving the theorem.

\end{proof}

\bibliographystyle{apa}
\bibliography{references}


\section{Missing Proofs}\label{apx:MissingProofs}
We recall Lemma~\ref{lem:beta}.
\begin{lemmanonumber}
If $\O'$ is the optimum schedule on time window $W - \delta$, then $$f(\S_r)\geq (1-\frac1e) f(\O').$$
\end{lemmanonumber}
\begin{proof}
Let $\S'_{r}= \{\left(M_{1}, \alpha_{1}\right), \dots, \left( M_{r}, \alpha_{r}\right)\}$ be the set of configurations picked by the greedy algorithm before the update steps (11)-(12) in which $\alpha_r$ is reduced to $\beta_r:=W-\delta-\sum_{j=1}^{r-1}(\alpha_j+\delta)$ to obtain schedule $\S_r$. Note that $\beta_r$ could be negative, however, for now assume $\beta_r \geq 0$. For ease of notation we also define $\beta_i=\alpha_i$ for each $1\leq i\leq r-1$. Thus $\S_{r}= \{\left(M_{1}, \beta_{1}\right), \dots, \left( M_{r}, \beta_{r}\right)\}$. We also let $\S_i$ to be the scheduled formed by picking the first $i$ configurations in $\S_r$. We now show the following claim.

\begin{claim}\label{cl:subineq}
For any configuration $(M_i,\beta_i)$ picked by the greedy algorithm in schedule $\S_r$ at any $1\leq i\leq r$, we have
$$f_{\S_{i-1}}\left(\left(M_{i}, \beta_{i}\right)\right) \geq \frac{\beta_{i} + \delta}{W-\delta}\left( f\left(\O'\right) - f\left(\S_{i-1}\right)\right).
$$
\end{claim}
\begin{proof}
First let us concentrate on the case when $i<r$. Then $\beta_i=\alpha_i$. Note that since $M_{i}$ is a matching that maximizes $\frac{||\min\left(R_{i}, \alpha_{i}M_{i}\right)||}{\alpha_{i} + \delta}$, for any other $M \in \mathcal{M}\setminus \S_{i-1}$ and any $\alpha \in \mathbb{R}_{+}$ we can write
$$ \frac{f_{\S_{i-1}}\left(\left(M,\alpha\right)\right)}{\alpha + \delta} \leq \frac{f_{\S_{i-1}}\left(\left( M_{i}, \alpha_{i}\right)\right)}{\alpha_{i} + \delta}$$
or equivalently, for each $1\leq i\leq r$ and configuration $(M,\alpha)$, we have
\begin{equation}\label{first_main}
f_{\S_{i-1}}\left(\left(M, \alpha\right)\right) \leq \frac{\alpha + \delta}{\alpha_{i} + \delta} f_{\S_{i-1}}\left(\left(M_{i}, \alpha_{i}\right)\right).
\end{equation}

For any $1 \leq i \leq r$, consider the following
\begin{equation}\label{eq:submodular}
 f\left(\O'\right) - f\left(\S_{i-1}\right) \leq f\left(\O' \cup \S_{i-1}\right) - f\left(\S_{i-1}\right) = f_{\S_{i-1}}\left(\O'\right) \leq \sum_{\left(M, \alpha\right)\in \O'\setminus \S_{i-1}} f_{\S_{i-1}}\left(\left(M, \alpha\right)\right)
\end{equation}
The last inequality comes from the submodularity of the function. Summing Inequality~\eqref{first_main} over all configurations in $\O'\setminus \S_{i-1}$ and using that the $\O'$ has a time window $W-\delta$, we obtain that
$$\sum_{\left(M, \alpha\right)\in \O'\setminus \S_{i-1}} f_{\S_{i-1}}\left(\left(M,\alpha\right)\right) \leq \frac{W-\delta}{\alpha_{i} + \delta} f_{\S_{i-1}}\left(\left( M_{i}, \alpha_{i}\right)\right). $$
Combining the above inequality with Inequality~\eqref{eq:submodular}, we obtain
\begin{equation}\label{eqn:gain}
f_{\S_{i-1}}\left(\left(M_{i}, \alpha_{i}\right)\right) \geq \frac{\alpha_{i} + \delta}{W-\delta}\left( f\left(\O'\right) - f\left(\S_{i-1}\right)\right).
\end{equation}

Thus if $i<r$, the claim follows since we have $\beta_i=\alpha_i$. 
When $i=r$, first observe that since the data sent along a single matching is a concave function of the time it is used in a configuration, we have that
\begin{eqnarray*}
f_{\S_{r-1}}\left(\left(M_{r}, \beta_{r}\right)\right) &\geq& \frac{\beta_{r} + \delta}{\alpha_r+\delta}f_{\S_{r-1}}\left(\left(M_{r}, \alpha_{r}\right)\right)\\
& \geq& \frac{\beta_{r} + \delta}{\alpha_r+\delta}\frac{\alpha_{r} + \delta}{W-\delta}\left( f\left(\O'\right) - f\left(\S_{r-1}\right)\right)\geq \frac{\beta_{r} + \delta}{W-\delta}\left( f\left(\O'\right) - f\left(\S_{r-1}\right)\right)
\end{eqnarray*}
This completes the proof of the claim.
\end{proof}

First note we can write the following equality:
$$ f\left(\O\right) - f\left(\S_{r}\right) = f\left(\O\right) - f\left(\S_{r-1}\right) - f_{\S_{r-1}}\left(\left(\M_{r},\beta_{r}\right)\right)$$
We will now derive the approximation factor. We have that
$$ f\left(\O\right) - f\left(\S_{r-1}\right) - f_{\S_{r-1}}\left(\left(\M_{r},\beta_{r}\right)\right) \leq f\left(\O\right) - f\left(\S_{r-1}\right) - \frac{\beta_{r} + \delta}{W - \delta}\left( f\left(\O\right) - f\left(\S_{r-1}\right)\right)   $$
This will result in the following inequality.
$$ f\left(\O\right) - f\left(\S_{r}\right) \leq \left(f\left(\O\right) -f\left(\S_{r-1}\right)\right)\left(1-\frac{\beta_{r} + \delta}{W-\delta}\right) $$
Continuing for the remaining $r-1$ steps we will have:
$$ f\left(\O\right) - f\left(\S_{r}\right) \leq \left(f\left(\O\right) -f\left(\S_{0}\right)\right) \Pi_{i=1}^{r}\left(1 - \frac{\beta_{i} + \delta}{W-\delta}\right) $$
Now using $1-x\leq e^{-x}$, we can write:
$$ f\left(\O\right) - f\left(\S_{r}\right) \leq f\left(\O\right)e^{-\sum_{i=1}^{r}\left(\frac{\beta_{i} + \delta}{W-\delta}\right)} $$
Since $\sum_{i=1}^{r} (\beta_{i} + \delta) \geq W - \delta$,
$$f\left(\S_{r}\right) \geq \left(1-\frac{1}{e}\right)f\left(\O\right)$$
Thus concluding the theorem for $\beta_r \geq 0$. Note that if $\beta_r < 0$, the whole argument of this section still holds without considering $\beta_r$ and the last configuration. This is because $\sum_{i=1}^{r-1} (\beta_{i} + \delta) \geq W - \delta$ holds without the last configuration if $\beta_r < 0$. Notice that in this case the algorithm will drop the last configuration.
\end{proof}

\section{Bi-Criteria for Online Variant}\label{apx:BicriteriaOnline}
We present an example that shows that a bi-criteria approximation is needed in the online variant of the circuit switch scheduling problem.
Given $\delta$ choose any time window $W$ such that $W\geq \delta+1$.
The input demand matrices are all zeros until the last time step $W$ in which the adversary injects a demand matrix $D$ that corresponds to a specific matching $M$ with demands of $1$ on all edges of $M$ and a demand of $0$ for all edges not in $M$.

The optimal solution knows $M$ in advance, and since $W\geq \delta +1$, it can spend $\delta$ time steps to switch to $M$ and fully satisfy $D$ in the last time step.
If we assume the online algorithm is deterministic, then let $M_{\text{alg}}$ be the matching that the online algorithm is configured to at the beginning of time step $W$ (if at all).
The adversary can choose a matching $M$ that is disjoint from $M_{\text{alg}}$.
Thus, no matter if the online algorithm changes the matching $M_{\text{alg}}$ or not, it cannot transmit even a single unit of demand.
If we assume the online algorithm is random, the adversary can choose a random uniform matching $M$.
Thus, no matter which matching the online algorithm chose, it transmits a single unit of demand in expectation whereas the optimal solution can satisfy all $n$ units of demand.
Hence, we conclude that without a bi-criteria guarantee any online algorithm cannot achieve any non-negligible competitive ratio.

\end{document}